\documentclass{amsart}

\usepackage{siunitx}
\usepackage{authblk}
\usepackage{subcaption}
\usepackage{amssymb}
\usepackage{amsmath}
\usepackage{color}
\usepackage{graphicx}
\usepackage{comment}
\usepackage{float}
\usepackage{hyperref}

\usepackage[title]{appendix}

\usepackage{changes}
\definechangesauthor[name={Christophe Pouzat}, color=orange]{xtof}
\definechangesauthor[name={Ludmila Brochini}, color=purple]{ludb}

\usepackage{soul,color}
\soulregister\cite7
\soulregister\ref7
\soulregister\pageref7

\usepackage{algorithm}
\usepackage{algcompatible}

\def\one{{\bf 1}\hskip-.5mm}

\def\E{{\mathbb E}}
\def\P{{\mathbb P}}
\def\R{{\mathbb R}}
\def\Z{{\mathbb Z}}
\def\V{{\mathbb V}}

\def\N{{\mathbb N}}

\def\F {{\mathcal F}}
\def\AA{{\mathcal A}}
\def\B{{\mathcal B}}
\def\Y{{\mathbb Y}}
\def\L {{\Lambda}}
\def\s {{\sigma}}

\usepackage[square,sort,comma,numbers]{natbib}

\newtheorem{theo}{Theorem}
\newtheorem{proposition}{\indent Proposition}

\newtheorem{defin}{\indent Definition}

\newtheorem{assumption}{Assumption}

\expandafter\def\expandafter\UrlBreaks\expandafter{\UrlBreaks
  \do\a\do\b\do\c\do\d\do\e\do\f\do\g\do\h\do\i\do\j%
  \do\k\do\l\do\m\do\n\do\o\do\p\do\q\do\r\do\s\do\t%
  \do\u\do\v\do\w\do\x\do\y\do\z\do\A\do\B\do\C\do\D%
  \do\E\do\F\do\G\do\H\do\I\do\J\do\K\do\L\do\M\do\N%
  \do\O\do\P\do\Q\do\R\do\S\do\T\do\U\do\V\do\W\do\X%
  \do\Y\do\Z}

\makeatletter
\g@addto@macro{\endabstract}{\@setabstract}
\newcommand{\authorfootnotes}{\renewcommand\thefootnote{\@fnsymbol\c@footnote}}%
\makeatother
\begin{document}
\begin{center}
  \LARGE 
  Estimation of neuronal interaction graph from spike train data \par \bigskip
\bigskip

  \normalsize
  \authorfootnotes
  Brochini, L.\footnote{To whom correspondence should be addressed: brochini@usp.br}\textsuperscript{1}, Galves, A.\textsuperscript{2},
  Hodara, P.\textsuperscript{1}, Ost, G.\textsuperscript{1} and
 Pouzat, C.\textsuperscript{2} \par \bigskip

  \textsuperscript{1}Universidade de S\~ao Paulo, S\~ao Paulo,  Brazil \par
  \textsuperscript{2}Paris-Descartes University and CNRS UMR 8145, Paris, France\par \bigskip

\end{center}

\vspace{0.3cm}

\begin{abstract}
{One of the main current issues in Neurobiology concerns the understanding of interrelated spiking activity among multineuronal ensembles and differences between stimulus-driven and spontaneous activity in neurophysiological experiments. Multi electrode array recordings that are now commonly used monitor neuronal activity in the form of spike trains from many well identified neurons. A basic question when analyzing such data is the identification of the directed graph describing ``synaptic coupling'' between neurons. In this article we deal with this matter working with a high quality multielectrode array recording dataset~\cite{pouzat2015} from the first olfactory relay of the locust, \emph{Schistocerca americana}. From a mathematical point of view this paper presents two novelties. First we propose a procedure allowing to deal with  the small sample sizes met in actual datasets. Moreover we address the sensitive case of partially observed networks. Our starting point is the procedure introduced in \cite{dglo}. We evaluate the performance of both original and improved procedures through simulation studies, which are also used for parameter tuning and for exploring the effect of recording only a small subset of the neurons of a network. }     
\end{abstract}
\textbf{keywords:} \small{ Statistical model selection,	interaction graph, biological neural networks, interacting chains of variable memory length, spike train}
\section{Introduction}

The issue of how neuronal networks are built to produce physiologically meaningful activity has been a long time concern in Neurobiology. While brain imaging techniques have been improving our understanding of how brain areas are activated to generate behavior, at the network level lies the challenge to understand how populations of neurons operate concurrently to process information that can be attributed to function of macroscopic structures. As multiunit recording techniques -- such as multi-electrode arrays  -- evolve concomitantly with computational power \citep{StevensonKording2011}, neural data analysis tools should also advance to be able to address current issues in Neurobiology. 

The statistical analysis of spike trains considered as what would now be called \emph{point process} can be traced back at least to \citet{Hagiwara:1954} and \citet{Fitzhugh_1958} who worked with single neuron data and to \citet{RodieckEtAl_1962} who worked with ``multi-variate'' spike trains (from 4 neurons). This work was reviewed by \citet{PerkelEtAl_1967} and \citet{PerkelEtAl_1967b} who also introduced in the spike train analysis field most of the modern point process terminology while focussing on mainly descriptive approaches.

A major subsequent progress was made by \citet{Brillinger_1976} and \citet{BrillingerEtAl_1976} who introduced identification methods transforming cross-correlograms and cross-intensity diagrams into quantitative tools. 
A subsequent progress was again made by \citet{brill88} in which an explicit conditional intensity model taking into account the synaptic interactions among neurons was proposed -- see also \citet{ChornoboyEtAl_1988} for a work done independently and \citet{Brillinger_1992} for a review.

The framework proposed by these papers can be summarized as follows: the activity of a population of neurons can be described as a system of interacting point processes -- either in discrete or continuous time. Such a system is defined by two elements, the first one is a directed graph having the neurons as vertices with an edge going from neuron $i$ to neuron $j$ if $i$ is \emph{presynaptic} to $j$. The second element is a set of functions indexed by the neurons mapping the past history of the system into the set of positive real numbers. In a continous time description, these functions tell at each time $t$ the spiking rate of each neuron, while in discrete time, they give the probability of spiking of each neuron at each time step. Several papers have  recently adopted this framework \cite{TruccoloEtAl_2005,OkatanEtAl_2005,Reynaud-Bouret+:2014}.

In alignment with these studies, we consider the class of models  introduced by Galves and L\"ocherbach in \cite{GalEva:13} -- and further studied in  \cite{Errico:14,Evafou:14,AG:14,Duarte_Ost:15,PierreEva:14, RobTou:14, Karina:15} -- which is built upon the simple assumption that the variable representing the membrane potential of each neuron is reset every time it spikes, which is biologically plausible.

In this class of models, the probability of a neuron to spike increases monotonically with its membrane potential, which evolves as follows. Whenever $i$ fires, its membrane potential is reset to a resting value. Moreover, when a presynaptic neuron $j$ fires, the membrane potential of neuron $i$ gains $W_{j\rightarrow i}$. At each time step, the neuron is subject to a leakage effect that drives the membrane potential to a resting state. Therefore, at a particular instant, the membrane potential of a given neuron is a function of the spiking history of all neurons presynaptic to it since the last time $i$ spiked.

Assuming the model description above, we study the statistical selection procedure of influence graph estimation proposed by \cite{dglo} that can be described as follows. For each neuron $i$ in the sample, we consider all other neurons as candidates to \textit{presynaptic}. We then estimate the probability of $i$ to spike given the spiking history of all other neurons since $i$ last spiked. For each neuron $j\neq i ,$ we then use a measure of sensitivity to determine if the conditional spiking probability is affected by changes in $j$ spiking activity. If this sensitivity measure is statistically small, we exclude candidate neuron $j$ from the set of neurons presynaptic to $i$.

The estimator we propose here contrasts with many previous studies that focus on pairwise statistical relationships between spike trains (\cite{Aertsen1989, Moore70, Gerstein1969, PerkelEtAl_1967b, Brillinger_1992, BrillingerEtAl_1976}).  Although these methods are still  widely used among neurobiologists for providing useful estimations of functional connectivity, they do not allow for discrimination of direct and indirect effect of a neuron activity over another. A more sophisticated related concept is that of ``effective connectivity" introduced by \cite{Aertsen1991} relating functional connectivity to a model of causal relation between neurons in a network  \cite{Horwitz2003}. The effective connectivity of a network is the simplest circuit that explains the temporal relationships between recorded neurons given a model of neural activity \cite{Friston1994},\cite{Friston2011}. The graph estimator considered in this work  provides more than just the functional connectivities, but infers, by construction, the effective connectivity of the network.  The consistency of the procedure is proved in \cite{dglo} and we verify that it recovers the real connectivity matrix for simulated data.

Furthermore, we address the sensitive and often avoided  issue of the influence of the unobserved neurons over the observed neuronal ensemble. Previous works typically ignore the existence of foreign input to the observed network or, as Duarte et al. \cite{dglo} build their model over the assumption (Assumption 6 of \cite{dglo} ) that all unobserved presynaptic neurons exert a negligible effect over the observed neurons. 

Unfortunately this assumption is not reasonable for the vast majority of experimental data in Neuroscience, especially when acquired through the multi-electrode array technique \textit{in vivo}, where the number of recorded neurons -- even when it is in the order of hundreds -- is typically many orders of magnitude smaller than the number of neurons presynaptic to the set of recorded neurons. This matter is particularly relevant to experimental data we analyse. Specifically, our data consists of multi-electrode array recordings of the olfactory pathway of the locust \emph{Schistocerca americana} -- an animal model for the study of chemosensory information processing. In the antennal lobe of this animal, which is the first olfactory relay, there is a cell type that fires very low amplitude action potentials that cannot be recorded even in close proximity to the electrode. The present work deals with this matter providing some theoretical results concerning the estimator for partially observed neuronal networks.

The Paper is organized as follows. In section \ref{sec:expmeth} we briefly describe experimental methodology and data.  In section \ref{sec:model}, we describe the model, the estimation procedure and propose a pruning algorithm in order to deal with the situation where the sample size is insufficient in the case where all neurons are observed. In section \ref{sec:simulations} we study the properties of the estimation procedures on simulations. We also detail the particular case of partial observation of the network. In section \ref{sec:truedata} we apply the procedures on a dataset recorded of the locust antennal lobe.  All theoretical results are presented in the Appendix \ref{app:maths} .

\section{Experimental  Methodology}\label{sec:expmeth}

The data--\emph{spike trains}--used in this manuscript were obtained
from the first olfactory relay, \emph{the antennal lobe}, of locusts,
\emph{Schistocerca americana}.  The antennal lobe contains two different cell types: $\approx$ 830 \emph{projection neurons} that are excitatory (cholinergic) and fire sodium dependent, propagated, action potentials (that is, classical action potentials) and also $\approx$ 300 \emph{local neurons} that are inhibitory (gabaergic) and fire calcium dependent, \emph{local}, action potentials that cannot be detected extracellularly. Therefore, extracellular recordings used to get the data \emph{do not catch the calcium spikes of the local neuron; a single cell type is seen: the projection neuron} \cite{Laurent_1996}. The recording method is described in
Pouzat et al (2002) \cite{pouzat2002}. 

 Briefly, extracellular recordings were performed \emph{in vivo} using a multi-electrode array (Now  commercialized by {\texttt{Neuronexus}}). Data were
acquired continuously (4 to 16 recording channels sampled at 15 kHz and
filtered with a band-pass filter between 300 and 5000 Hz) and stored on
a hard drive for subsequent analysis. These data are now publicly and
freely available from \texttt{Zenodo} (\url{https://zenodo.org/record/21589\# \kern-1ex .WQoUCx1Jlz9}) \cite{pouzat2015} where a description of each
experiment can also be found. The data are stored in  HDF5
format and the one or two files belonging to each experiment contain as
metadata the transcript of the experimentalists' lab book (\url{https://support.hdfgroup.org/HDF5/}).

The spike trains analyzed in the present manuscript were obtained from the ``raw
data'' found on zenodo after a procedure called {\texttt{Spike sorting}. Briefly, this procedure consists in detecting action potential or \emph{spike} from the raw data, determining how many different neurons are
contributing to the recording (how many different spike shapes are
observed) before classifying each detected event as belonging to one of
the active neurons. This procedure is not neutral, therefore the
analysis leading to the spike trains used in this manuscript is
described step by step on a dedicated GitHub
repository (\url{https://christophe-pouzat.github.io/zenodo-locust-datasets-analysis/}). On this repository the reader can find for each experiment
and for each tetrode (group of four recording electrodes) the full
script with comments and figures describing the sorting procedure (implemented in \texttt{R}) and allowing them to judge the quality of the obtained spike trains. Here we use only spike trains from well isolated neurons.

These so called spike trains, can be transformed into a symbolic sequence that represents the process in discrete time: the total duration of the recording is divided into small windows (with the order of milliseconds) and a symbol $1$ is attributed to each window that corresponds to the occurence of a spike and, conversely, a symbol $0$ is attributed to the absence of a spike in a window. Once this representation of neural behavior is obtained from data, it is possible to perform statistical model selection in the class of models proposed by Galves and Locherbach (\cite{GalEva:13}) by applying the interaction graph estimator procedure proposed in \cite{dglo}.

\section{Statistical model selection}\label{sec:model}

\subsection{A stochastic model for a network of interacting neurons}
\label{sec:GL}

In what follows,  $I$  denotes a finite set and $W_{j \to i}\in \R$ with $i,j\in I$, a collection of real numbers such that $W_{j \to j } = 0 $ for all $j.$ Moreover, for each $i\in I$, $\varphi_i : \R  \to [ 0, 1 ]$ is a  non-decreasing measurable function and $g_i=(g_i(t))_{t\geq 0}$ is a sequence of strictly positive real numbers. 
In neuroscience terms we call: $I$ the set of neurons; $W_{j \to i } $ the {\it synaptic weight of neuron $j$ on neuron $i,$}; $\varphi_i $ the  {\it spike rate function} of neuron $i$;  $g_i$ the {\it leakage function} of neuron i.

We shall consider a time-homogeneous stochastic chain $(X_t )_{ t \in \Z }$ taking values in $ \{ 0, 1 \}^I ,$ defined on a suitable probability space $ ( \Omega , \AA , \mathbb{P} )$. The configuration of the network at $t \in \Z$ is given by $X_t=(X_t(i),i\in I)$; for each $i\in I$,  $ X_t (i) = 1 $ if neuron $i$ spikes at time $t$ and $X_t(i) = 0 $ otherwise.

For notational convenience, we shall write $X_s^t(A)=(X_h(i),s\leq h\leq t,i\in A)$ for any set $A\subset I$ and indexes $-\infty\leq s\leq t< \infty$. When $A=I$ we shall write simply $X_s^t$ rather than $X_s^t(I)$. 
For each $i \in I$ and $t \in \Z $, let
\begin{equation}
\label{def:0}
L_t^i = \sup \{ s\leq t  : X_s (i) = 1  \}, 
\end{equation} 
be the last spike time of neuron $i$ before time $t.$ Here, we adopt the convention that  $\sup \{ \emptyset  \}=-\infty.$

The stochastic chain $(X_t )_{ t \in \Z }$ has the following dynamics. For each $t \in \Z $ and $ a_i \in \{ 0, 1 \} ,i\in I,$ it holds $\P-$almost surely that
\begin{equation}
\label{def:1}
\P(X_{t+1}(i) =a_i,i\in I|X^{t}_{-\infty})=\prod_{i\in I}\P(X_{t+1}(i)=a_i |X^{t}_{-\infty}). 
\end{equation}
The identity \eqref{def:1} means that, conditionally on the past history $X^{t-1}_{-\infty}$, the random variables $X_t(i),i\in I,$ are independent. 
Moreover, we assume that $\P-$almost surely for each $i\in I$ and $t\in \Z$, 
\begin{equation}
\label{def:2}
\P(X_{t+1}(i)=1|X^{t}_{-\infty})= \varphi_i\Big( \sum_{j\in I}W_{j\to i}\sum_{s=L_{t}^i+1}^{t} g_j(t-s)X_s(j)\Big)
\end{equation}
whenever $L_{t}^i<t$ and 
$$
\P(X_{t+1}(i)=1|X^{t}_{-\infty})= \varphi_i (0 ),
$$
otherwise. 

For each  $i\in I$, let 
\begin{equation}
\label{def:int_neigh_of_i}
V_i=\{j\in I\setminus\{i\}: W_{j\to i}\neq 0\}
\end{equation} 
be the set of presynaptic neurons of neuron $i$; it will be called the {\it interaction neighborhood} of neuron $i$.
The identity \eqref{def:2} means that, on the event $\{L^i_t< t\}$, 
the probability of neuron $i$ to spike at time $t+1$ depends on the past history $X^{t}_{-\infty}$ only through $X^{t}_{L^i_{t}+1}(V_i).$ 
Moreover, immediately after spiking, neuron $i$ will spike (or not) according a Bernoulli random variable with parameter $\varphi_i(0)$ regardless the past history.
Notice that the random variables $L^i_t$ introduce a  structure of variable-length memory in the model. 

We define the set  $\Omega^{adm} $  of \textit{admissible pasts} as follows  
\begin{equation}
\label{def:adm_set}
\Omega^{adm}=\left\{x  \in \{0,1\}^{I\times \{\ldots,-1,0\}}: \forall \ i\in I, \ \exists \ \ell_i\leq 0 \ \mbox{with} \  x_{\ell_i}(i)=1 \right\} .
\end{equation}
Notice that when $X_{-\infty}^{0}=x\in \Omega^{adm} $, we have $L^i_0>-\infty$ for all $i\in I$. As a consequence, for each $i\in I$,
$$
\sum_{j\in I}W_{j\to i}\sum_{s=L_{0}^i+1}^{0} g_j(-s)X_s(j)<\infty, 
$$ 
implying that the transition probability $\P(X_{1}(i)=1|X^{0}_{-\infty}=x)$ is well-defined. By induction, for each $t\geq 1$, the transition probabilities \eqref{def:2} are also well-fined. Therefore, the existence of the stochastic chain $(X_t)_{t \in \Z }$, starting from  $X_{-\infty}^{0}=x\in \Omega^{adm} $, follows immediately.
Notice that  we do not assume stationarity of the chain.

\subsection{Statistical inference.}
\label{Subsec:statis_inference}
The aim of this section is to present the statistical procedure proposed in \cite{dglo} to estimate, for a given neuron $i \in I$, the interaction neighbourhood $V_i$ from the data. 
In what follows it will be assumed that $V_i \neq \emptyset$ for all $i \in I.$ 
To present the estimation procedure we shall introduce further notation.

For any subset $A\subseteq I$ and integer $\ell\geq 1$, we define
$
C^{A,\ell}= \{ 0, 1\}^{ \{ - \ell , \ldots , -1\} \times A}.
$ 
Elements of $C^{A,\ell}$ are called local pasts.
For any integers $\ell\geq 1$ and $t \in \Z, $ subset $A\subseteq I$ and local past  $w\in C^{A,\ell}$, we shall write $ X^{t-1}_{t- \ell } ( A) = w $, if $ X_{t-s } (j ) = w_{-s} (j ) , $ for all $  1 \le s \le \ell $ and for all $ j \in A.$ 

Let $X_1(F),\ldots,X_n(F)$, with $F \subset I$, be a sample produced by a version of the stochastic model $(X_t)_{t \in \Z}$ compatible with \eqref{def:1} and \eqref{def:2}. Here, $F$ is the set of neurons whose spike activity has been recorded and $n$ is the length of the time interval during which this set of neurons has been observed. 

The estimation procedure is defined as follows.
For each integer $n\geq 3$, neuron $i\in F,$ local past $w\in  C^{F\setminus\{i\},\ell}$ outside of $i$  with $1\leq \ell\leq n-2$ and symbol $a\in\{0,1\}$, we define 

\begin{equation*}
N_{(i,n)}(w,a)=
\sum_{t=\ell+2}^{n}{\one}{\{ X_{t-\ell-1}^{t-1}(i)=10^{\ell},X_{t-\ell}^{t-1}(F \setminus\{i\})=w,X_t(i)=a\}}.
\end{equation*}
The random variable $N_{(i,n)}(w,a)$ counts the number of occurrences of $w$ followed or not by a spike of neuron $i$ ($a=1$ or $a=0,$ respectively) in the sample $X_1 (F) ,\ldots, X_n (F) ,$ when the last spike of neuron $i$ has occurred $\ell + 1$ time steps before in the past. 

For a fixed local past $w\in  C^{F\setminus\{i\},\ell}$ outside $i\in F$, we define then the empirical probability of neuron $i$ having a spike at the next time step given $w$ by
\begin{equation}
\label{def:trans.prob.emp}
\hat{p}_{(i,n)}(1|w)=\frac{N_{(i,n)}(w,1)}{ N_{(i,n)}(w)},
\end{equation}
when $N_{(i,n)}(w)=N_{(i,n)}(w, 0 ) +N_{(i,n)}(w, 1 )>0.$ 

For any fixed parameter $\xi\in (0,1/2)$, we consider the following set
\begin{equation}\label{eq:xi}
\mathcal{T}_{(i,n)}= \Big\{w\in  \bigcup_{\ell =1}^{n-2}  C^{F\setminus\{i\},\ell} :N_{(i,n)}(w)\geq n^{1/2+\xi} \ \Big\}.
\end{equation} 

We write $|w| = \ell$ whenever $w\in C^{F\setminus\{i\},\ell}. $
If $ v, w $ both belong to $C^{F\setminus\{i\},\ell} $  we write 
$$ 
v_{ \{j \}^c }   = w_{ \{ j \}^c }     \mbox{ if  and only if }  v_{-\ell }^{-1}\big(F\setminus\{j\}\big)=
w_{-\ell }^{-1}\big(F\setminus\{j\}\big)  .
$$
In words, the equality $v_{ \{j \}^c }   = w_{ \{ j \}^c }$ means that  $v$ and $w$ coincide on all but the $j$-th coordinate.

Finally, for each $w\in\mathcal{T}_{(i,n)}$ and for any $ j \in F \setminus \{i\} $ we define the set
\begin{equation*}
\mathcal{T}^{w,j}_{(i,n)}= \Big\{v \in \mathcal{T}_{(i,n)}: |v|=|w|,  v_{ \{j \}^c }   = w_{ \{ j \}^c }  \Big\}
\end{equation*}
and introduce the \textit{measure of sensitivity}
\begin{equation}\label{eq:delta}
\Delta_{(i,n)}(j)=\max_{w\in \mathcal{T}_{(i,n)}}\max_{v \in \mathcal{T}^{w,j}_{(i,n)}}|\hat{p}_{(i,n)}(1|w)-\hat{p}_{(i,n)}(1|v )|.
\end{equation}

For any positive threshold parameter $\varepsilon >0,$ the estimated interaction neighborhood of neuron $i \in F,$ at accuracy $\varepsilon,$ is defined as

\begin{equation}\label{eq:defestim}
\hat V_{(i,n)}^{(\varepsilon)} = \left\{ j \in F \setminus \{i\} : \Delta_{(i,n)}(j) > \varepsilon \right\}.
\end{equation}
In what follows, whenever $j\in \hat V_{(i,n)}^{(\varepsilon)}$,    we write $j\to i$ to denote the connection from $j$ to $i$. 
When $V_i\subset F$ and some other assumptions on the model hold, the estimated interaction neighborhood $\hat V_{(i,n)}^{(\varepsilon)}$ is shown to be {\it strong consistent}, as shows Theorem \ref{prop:noprune}. By strong
consistency we mean that $\hat V_{(i,n)}^{(\varepsilon)}$ equals to $V_i$ eventually almost surely as $n\to\infty$. 
For cases in which we do not necessarily assume that $V_i\subset F$, the strong consistency of $\hat V_{(i,n)}^{(\varepsilon)}$ can no longer be guaranteed (see Section \ref{sec:simulations} for details). In these cases, the estimated interaction neighborhood is only able to exclude, from the observed set $F$, neurons which are functionally independent of the neuron $i$, as stated in Proposition \ref{prop:partial_obs}
 

\begin{defin}
A pair of neurons $i,j\in I$ with $i\neq j$ are said to be functionally independent if there is no path of neurons $k_1,\ldots, k_m\in I$ with $i=k_1$ and $j=k_m$ such that for all $1\leq p\leq m-1$, either $W_{p\to p+1}\neq 0$ or $W_{p+1\to p}\neq 0$. 
\end{defin}

The theoretical results concerning the estimation procedure, including Theorem \ref{prop:noprune} and Proposition \ref{prop:partial_obs}, will be discussed in Appendix   \ref{sec:consitency}.

\subsection{Iterative pruning procedure}\label{subsec:pruning}

Since the estimator is well defined only on events of the type $ \bigcap_{j\in F: j\neq i} E_{i,j}^n,$ where $E_{i,j}^n := \left\{\exists w \in \mathcal{T}_{(i,n)}: \mathcal{T}_{(i,n)}^{w,j} \neq \emptyset \right\},$ we propose an iterative pruning procedure to deal with cases where this event is not realized.  For neuron $j\in F\setminus\{i\}$ for which $E_{i,j}^n$ is not realized the connection $j\to i$ will be called inconclusive. This may occur when the sample size is small.  When $\mathcal{T}_{(i,n)} = \emptyset$, all connections leading to neuron $i$ are considered inconclusive. In the case where $\mathcal{T}_{(i,n)} \neq \emptyset$, a connection $j\to i$ is considered inconclusive if $\mathcal{T}_{(i,n)}^{w,j} = \emptyset $ for all $w \in \mathcal{T}_{(i,n)}$ .

The pruning procedure is described as follows.  If there exist $j \in F\setminus\{i\} $ such that $E_{i,j}^n$ is not realized and $k \in F \setminus \{i,j\}$ such that $E_{i,k}^n$ is realized and $k \notin \hat{V}^{(\epsilon)}_{(i,n)}$, we say that the pruning condition is fulfilled. If so, we compute $\hat{V}^{(\epsilon)}_{i,n}$ considering the set $F \setminus \{i,k\}$ instead of $F\setminus\{i\}.$ This step is repeated as long as the pruning condition is fulfilled.
As we will show in the Appendix \ref{sec:pruningproof}, this iterative pruning procedure conserves the consistency of the estimation \ref{eq:defestim}. Its formal definition is given now as a pseudo-code.

\begin{algorithm}[H]
\caption{ Iterative pruning procedure to estimate $V_i$}
\small
\begin{algorithmic}[1]
\State {\bf Input:} The parameters $\xi$ and $\epsilon$, the observable set $F$ and the sample $X_1(F),\ldots, X_n(F)$ .
\State {\bf Output:} The selected set $\hat{V}_i$ for each $i\in F$.
\State \textit{Initial values:} $G_i\leftarrow F\setminus\{i\}$;\\
$ \mathcal{T}_i \leftarrow \{w \in \cup_{\ell=1}^{n-2} \{0,1\}^{\{-\ell,\ldots,-1\}\times G_i} : N_i(w) \geq n^{1/2+\xi} \}$; $ \mathcal{I}_i \leftarrow \emptyset$; $\hat{V}_i \leftarrow \emptyset $ 
\IF {$\mathcal{T}_i =\emptyset$}
\State $I_i\leftarrow G_i$ and $\hat{V_i}\gets \emptyset$
\ELSE 
\FOR {$j\in G_i $ }	
\FOR { $w\in \mathcal{T}_i $ } 
\State Compute $\mathcal{T}^{w,j}_i$
\IF {$\mathcal{T}^{w,j}_i\neq \emptyset$} {$M^{w,j}_i\leftarrow \max_{s\in \mathcal{T}^{w,j}_i\ }|\hat{p}_i(1|w)-\hat{p}_i(1|s)|$}
\ENDIF
\IF {$M^{w,j}_i>\epsilon$} 
\STATE $\hat{V}_i \leftarrow \hat{V}_i\cup \{j\}$ 
\ENDIF				
\ENDFOR
\IF {$\cup_w \mathcal{T}^{w,j}_i=\emptyset$} 
\State $I_i\leftarrow I_i\cup\{j\}$
\ENDIF			

\ENDFOR
\IF{$G_i \neq \hat{V_i}\cup I_i$}		
\STATE Choose $k \in G_i \setminus \left( \hat{V_i}\cup I_i \right).$		
\STATE $G_i\leftarrow G_i \setminus \{k\}$	 and return to step 4.		
\ENDIF
\ENDIF

\State {\bf Return:} $\hat{V}_i$ and $I_i$.
\end{algorithmic}
\end{algorithm}

\section{Results on simulations}\label{sec:simulations}
Simulation and graph estimation procedures used to produce the results presented in this section were implemented in Python 3.0 and are publicly available online at \url{https://github.com/lbrochini/Graph-Estimation}.

\subsection{Searching for suitable $\varepsilon$ and $\xi$ parameter values}

In this section, we use simulated data in order to fix the parameters $\xi$ and $\varepsilon$ involved in the estimation procedure. Recall that $\xi$ is the parameter appearing in the definition of the set $\mathcal{T}_{(i,n)}$ in \eqref{eq:xi} and that $\varepsilon$ appears in the definition of $\hat V_{(i,n)}^{(\varepsilon)}$ in \eqref{eq:defestim}. The role of $\xi$ is to ensure that the observations contains enough repetitions of a given local past $w$ in order to define the empirical probability $\hat{p}_{(i,n)}(1|w).$ The parameter $\varepsilon$ can be seen as a significance threshold for the measure of sensitivity $\Delta_{(i,n)}(j).$

The simulated samples have a similar sample size to the experimental dataset analysed in section (\ref{sec:truedata}), i.e. $n=10^6$ and with the same number of neurons: 5.  The neural activity is simulated according to the dynamics described in  Eqs.  \ref{def:1} and \ref{def:2}. Synaptic weights $W_{j\rightarrow i}$ were arbitrarily distributed from 0 to 0.8 in this network for all possible pairs $(i,j)$. We use a geometric leakage $g_i(n)=\mu^n$ with parameter $\mu=0.5$ for any neuron $i$. We use a simple linear saturating firing function $\varphi_i(u)=\min(u+q_i,1)$, where $q_i=0.02$ for any neuron $i$.  A simulation study of this model under geometric leakage was previously done in \cite{Brochinietal:16}.

In figure \ref{fig:simxieps}, we give the results of the estimation procedure for different values of the parameters $\xi$ and $\varepsilon.$  For each couple $(\xi , \varepsilon)$ we present the result in a $5 \times 5$ matrix. For each line $j$ and column $i,$ the color of the square indicates the presence or absence of influence of neuron $j$ on neuron $i$ and the result of the estimation procedure. The color code is the following. Correct estimations are represented in black and white: black if $W_{j \to i} \neq 0$ and white if $W_{j \to i} = 0$. Incorrect estimations are represented in hatched cells. Hatched white cells correspond to false negatives, when $W_{j \to i} \neq 0$ but the estimator produced an absent connection. Hatched grey cells, on the other hand, indicate false positives, when $W_{j \to i} = 0$ but a connection was estimated to exist. Plain grey cells correspond to inconclusive results, a situation when the event  $E_{i,j}^n$ is not realized, where  $E_{i,j}^n := \left\{\exists w \in \mathcal{T}_{(i,n)}: \mathcal{T}_{(i,n)}^{w,j} \neq \emptyset \right\}$. This may happen due to the fact that the sample becomes relatively small as the cutoff parameter $\xi$ increases, in which case the procedure will produce a smaller number of valid events to be considered by the estimator.

As expected, low values of the sensitivity threshold $\varepsilon$ lead to more false positive whereas high values lead to more false negative. For this sample size, the estimation procedure correctly recovers the true connectivity graph for $\varepsilon=0.05$ and $\xi=0.001$ or $0.01.$

\subsection{Pruning }

The pruning procedure proposed in section \ref{subsec:pruning} can be applied in cases where there are inconclusive results. It consists in removing connections estimated as absent and obtaining a new estimation for the reduced subset.

Inconclusives can be typically attributed to small sample sizes and/or data sparsity. Evidently, if we increase 
the number of neurons or decrease sample size while maintaining the same parameter values of $\epsilon$ 
and $\xi$, we expect a larger number of inconclusive connections. This is precisely what we did to illustrate the utility of the pruning procedure:  we generated a sample of  GL network activity with 10 neurons and sample size of $n=\num{2e5}$, which is a larger number of neurons and smaller sample size as used in the previous section. All synapses have the same weight ($W=0.5$) and leakage and spontaneous activity parameters are set to $\mu=0.9$ and $q=0.06$. In the analysis we used parameter values  $\epsilon=0.05$ and $\xi=0.001$, determined in the previous section. 

In fact, the first estimation obtained prior to any pruning (shown in Fig. \ref{fig:pruning} A) produced a remarkably large number of inconclusives (grey cells).

After the first estimation, the pruning procedure is used to help reduce the amount of inconclusives. For each postsynaptic neuron $i$, all neurons which are identified by the estimator as not preysynaptic to $i$ are removed from the set of presynaptic candidates. Then the graph estimating procedure is repeated. The pruning and re-estimation is repeated while there are at least one inconclusive and one connection identified as null for the postsynaptic neuron $i$.

After the pruning procedure is performed for all postsynaptic neurons, we observe a dramatic improvement in the quality of the graph estimation (Fig. \ref{fig:pruning} B).  The final estimation correctly identifies all existing connections for this network. The effectiveness of the pruning procedure is due to the reduction in the number of presynaptic candidate neurons while maintaining the same sample size, leading to the improvement of the estimation quality.

\subsection{Connectivity graph estimation of observed neurons in the presence of unobserved synaptic paths: an empirical investigation}

Our results presented in the Appendix  \ref{App:partial_obs} guarantee that if there is no path between neuron j and i involving an unobserved neuron, then the estimator is not expected to produce a connection. What happens when this condition is not met? Evidently, it is possible for a specific pair of neurons to be disconnected and yet influence one another indirectly through a set of unobserved neurons to which they are both connected, configuring a path. In this case, the  estimator can produce a connection that in fact is a \textit{projected} connection of that path. 


As the estimator will not  consistently produce the true connectivity graph in this situation,  it is desirable to empirically investigate if  the graph estimation  is able to provide useful information about the true graph in the particular case where there is a \textit{known} path involving unobserved neurons.  In order to do so, we generate a sample of a simulated network of $N$ neurons where the true graph is known beforehand. Then, we apply the graph estimation for a subset of these neurons, representing the  observed neurons. Once the estimation is obtained, we compare the estimated connections of the subgraph of observed neurons with the complete true graph. 

Figure ~\ref{ProjMap} A depicts the schematic connections of a simple network of GL neurons with 10 neurons, with 7 synapses total, all excitatory with the same weight.
 A sample of size $\num{2e5}$ was generated with the dynamics introduced in Sec.~\ref{sec:GL}. We then computed the estimated graph for subsets of the network with size 3 and 4. All subgraph estimations for all possible subsets produced the following results:
 
\begin{itemize}
\item All connections identified as false were indeed false
\item All connections identified as true were either indeed true or a \textit {false positive due to a projected connection} 
\end{itemize}
 
This empirical study indicates that the graph estimation procedure can be applied to real \textit{in vivo} recordings of multiunit spike trains, as the connections can be attributed either to true connections due to an existent synapse or projected connections from a synaptic path of unobserved neurons.  
 
Now we  propose a procedure that, without assuming prior knowledge about the true connectivity, would allow us to obtain a connectivity estimation for the whole graph based on the estimations for all possible subgraphs. We propose to use subsets of size 3, that provides the smallest number of combinations and sample size required for the analysis.

The procedure is as follows. After computing the graph estimation for all subsets of size 3,  we compare all estimation outcomes for a single pair $(j,i)$ to verify whether a connection from $j$ to $i$ is identified in all, none or some of the subgraph estimations for all $N-2$ subsets that contain this specific pair.  If the connection is identified as absent for all subsets, then we consider it absent for the complete set of observed neurons. Conversely, if the estimator produces a connection $j \rightarrow i$ for all subsets, then we estimate that there is a direct connection  $j \rightarrow i$. Alternatively, if the estimator produces a connection $j \rightarrow i$ for some subsets but also identifies it as an absent connection for other subsets, then we consider that there is no direct connection.

Figure \ref{ProjMap} B depicts the results of this procedure applied to the aforementioned sample. In the figure, the cell at line $j$ and column $i$ describes the nature of influence that neuron $j$ has over neuron $i$. Cells depicted in black indicate a connection detected for each subset containing the concerned pair of neurons. White cells indicate  no subset containing the concerned pair of neurons produced a connection. Hatched cells represent pairs of neurons that produced a connection for some subsets but had an absent connection for some other subsets. 

We can compare figures  \ref{ProjMap} A and B to check if the recovered whole graph corresponds to true connections.  Indeed, black and white cells correspond to true and absent connections respectively, while hatched cells correspond to disconnected pairs in the true graph induced by a projection (dashed lines in fig \ref{ProjMap} A) due to a path with one or more other neurons.

By virtue of this result, we believe this procedure can be a good strategy to be applied to real data. The typical number of recorded neurons in an experiment has been doubling every seven years because of continuous advances in multiunit recording techniques \cite{StevensonKording2011}. This introduces a limitation for estimating the connectivity graph for all neurons simultaneously. Indeed, if the number of neurons grows too large with respect to the sample size the experimenter is able to obtain, the estimator should produce more and more inconclusive connections. Therefore, we propose the  use the procedure presented in this section to obtain a connectivity graph estimation from data recorded from large amounts of neurons.

\section{Results on a dataset recorded in vivo}\label{sec:truedata}

Here we present results of the estimated influence graph for a particular dataset that corresponds to a recording of about half an hour of spontaneous neural activity. Spike sorting procedure for this dataset can be found here : \url{https://christophe-pouzat.github.io/zenodo-locust-datasets-analysis/Locust_Analysis_with_R/locust20010217/Sorting_20010217_tetD.html}. Through this procedure we obtain spike trains of 5 well isolated neurons, each neuron presenting the order of $10^4$ total spikes in the sample.

In order to use the estimation procedure, we need to obtain a representation of the spike train in discrete time. We choose the largest binning window which produces less than $1\%$ of overlaps. By overlap we mean when two or more spike events of the same neuron occur in the same time window. This leads to a binning window of about 10 milliseconds.

We fix for $\xi$ and $\varepsilon$ the values that fitted the simulations, i.e. $\xi=0.001$ and $\varepsilon=0.05.$ We present in figure \ref{fig:truedata}A the result of the estimation procedure. The color code is the following: black  indicates we estimated that there is a connection $i\leftarrow j$, white indicates we estimated that there is no connection $i\leftarrow j$ and grey corresponds to an inconclusive.

Unfortunately the results are mostly inconclusives for  neurons 4 and 5, even with the pruning procedure described in section \ref{subsec:pruning}.

In order to validate this estimation procedure, we split the dataset in two parts and proceed to the estimation for each part. The results are given in figures \ref{fig:truedata} B and \ref{fig:truedata} C.

We can see that the estimation procedure gives us the same graph for the two different parts of the dataset, except for pairs 4 $\leftarrow$ 5 and 5 $\leftarrow$ 4 where we have inconclusive results when data is split.  As was already mentioned, the expected number of inconclusives should be very sensitive to sample size, so it is not surprising that two connections considered absent when the whole data is analyzed appear as inconclusives when sample size is reduced by half. Having this considered, we can conclude that there is an overall agreement between the graphs obtained, and say that the estimation obtained  is robust to data splitting for this dataset of projection neurons in spontaneous activity.

\section{Discussion}


This work presents a novel method to estimate the connectivity graph of a network of neurons given their spike train recordings by providing theoretical, empirical and applied results, based on a previous work by Duarte et al \cite{dglo}. 
Although the graph estimation procedure here proposed is not agnostic about the model, our core assumption is very simple and biologically plausible: that the neuron forgets all the process past to a spike.

In the discrete time framework, this reset assumption makes the model a system of stochastic chains with memory of variable length. This generalizes the original notion of context tree introduced by Riessanen \cite{Rissanen:1983}.  With the specific definition introduced in \cite{GalEva:13} this model can be seen as a system of leaky integrate and fire neurons with random threshold \cite{Brochinietal:16}. This assumption is also mathematically challenging as it makes the model non-continuous and without monotonicity properties.

This model has been intensively studied in past few years 
including Duarte et al. \cite{dglo} that performed a first study on graph estimation on this model. This previous work does not address, however, important and novel issues dealt in the present work, such as the influence of unobserved stimuli, assuming that, if they exist,  exert negligible effect over the observed neurons.

The issue of the ubiquitous presence of unobserved presynaptic stimuli, although often neglected, is pertinent to most problems of inference of neurons interrelationships from their recorded spike trains.  In  \cite{brill88}, Brillinger preforms an empirical study of the integrate-and-fire model with random threshold and, although the model addresses the question of interconnections in moderate sized networks providing biologicallly interpretable parameters, he does not consider the effect of hidden presynaptic activity.

There is an even more basic issue regarding what is considered connectivity among observed neurons. Traditionally, recorded neurons are considered to interact if they exhibit correlated firing patterns, typically inferred through  pairwise cross-correlograms \cite{PerkelEtAl_1967b} or joint peri-stimulus time histograms \cite{Gerstein1969}. A pair of neurons  with correlated activity may indicate the presence of a synaptic connection between them or, alternatively, an indirect synaptic pathway or a common drive from other unobserved neurons in the system. 

As it is not possible through neural activity recording techniques to positively determine if two neurons are synaptically connected, neurons with correlated activity are simply said to be functionally connected. A complementary measure of connectivity is that of \textit{effective connectivity} \cite{Aertsen1991,Friston1994,Friston2011}. Okatan and others \cite{OkatanEtAl_2005} with an interesting approach, present a similar concept as what they call optimal connectivity matrix.  As a matter of fact, our method provides an estimation of effective connectivity of the observed neural ensemble.

 Here we apply the estimator to the analysis of neural activity data recorded from the first olfactory relay\textemdash{}the antennal lobe\textemdash{}of a locust, \textit{Schistocerca americana}.  \textit{In vivo} antennal lobe recordings in insects consist an important paradigm to the study of olfaction, especially because olfactory systems anatomy  is well conserved across many species. Studies in locust, for instance,  have provided important insights into how odors are represented in the brain \cite{Laurent_1996}. Experimental protocols typically involve recording spontaneous  neural activity as well as stimulus-driven activity in response to the repeated  presentation of odors.

We believe the influence graph estimator here studied could be broadly used in the analysis of multiunit recordings to infer effective connectivity of recorded neurons, allowing for comparisons of the estimated influence graph when neurons are in spontaneous firing mode or responding to stimuli.

\section{Software}
\label{sec5}
All software and data pertinent to this work are publicly available online. Raw data from multielectrode array recordings from the locust \textit{Schistocerca americana} are available at zenodo (\url{https://zenodo.org/record/21589\# \kern-1ex .WQoUCx1Jlz9}, \cite{pouzat2015}) in hdf5 format. A thorough description of the spike sorting procedure implemented in \texttt{R} and a step by step analysis can be found at \url{https://christophe-pouzat.github.io/zenodo-locust-datasets-analysis/Locust_Analysis_with_R/locust20010217/Sorting_20010217_tetD.html}. All simulations and analysis routines were implemented in Python 3.0 and are available at \url{https://github.com/lbrochini/Graph-Estimation} along with the processed data and a guide for reproducing the figures presented in this manuscript.

\section*{Acknowledgments}

This work was produced as part of the activities of S\~ao Paulo Research Foundation (FAPESP) Research, Innovation and Dissemination Center for Neuromathematics (grant no.  2013/ 07699-0). L.B. also received Conselho Nacional de Desenvolvimento Cientfico e Tecnolgico (CNPq) support (grant no.  165828/2015-3) and FAPESP support (grant no. 2016/24676-1) .  P.H. received FAPESP support (grant no.  2016/17655-8). G.O. received FAPESP support grant no. 2016/17789-4.

\begin{appendix}

\section{Theoretical results}\label{app:maths}

\subsection{Consitency of the estimation}\label{sec:consitency}

To prove the consistency of our estimator we impose the following.
\begin{assumption}\label{ass:4}
For all $ i \in I, $  $\varphi_i \in C^1 ( \R, [0, 1]  ) $ is a strictly increasing function. Moreover, there exists a $p_{* } \in ]0, 1 [ $ such that for all $i\in I$ and $u\in \R$
$$
p_* \leq \varphi_i(u)\leq 1-p_{*}.
$$  
\end{assumption}

For each $i\in I$, write 
\begin{equation}
\label{def:Kil}
K_{i}=\left[\ \sum_{j\in V^{-}_i}W_{j\to i}g_j(0),\sum_{j\in V^{+}_i}W_{j\to i}g_j(0)\right],
\end{equation}
where $V^{+}_i=\{j\in V_i: W_{j\to i}>0\}$ and $V^{-}_i=\{j\in V_i: W_{j\to i}<0\}$ and define
\begin{equation}\label{eq:mi}
m_i=\inf_{u\in K_{i}}\left\{\varphi'_i(u)\right\}\inf_{j\in V_i}\left\{|W_{j\to i}|g_j(0)\right\}.
\end{equation}

The consistency theorem is as follows. 
\begin{theo}[Theorem 1 of \cite{dglo}]\label{prop:noprune}
Let $X_1 (F) ,\ldots, X_n(F)$ be a sample produced by a stochastic chain $(X_t)_{t\in\Z}$ compatible with \eqref{def:1} and \eqref{def:2}, starting from $X^{0}_{-\infty}=x $ for some fixed $ x\in \Omega^{adm}$. Under Assumption \ref{ass:4}, for any $i\in F$ satisfying $V_i \subset F$ the following holds.
\\
1. {\bf (Overestimation).} \label{thm:2I} For any $j\notin V_i$, we have that for any $\epsilon>0,$
$$
P\Big(j\in \hat{V}^{(\epsilon)}_{(i,n)}\Big)\leq  p_n^o := 4n^{3/2-\xi}\exp\left\{-\frac{\epsilon^2 n^{2\xi}}{2}\right\}.
$$
2. {\bf (Underestimation).} \label{thm:2II} 
The quantity $m_i$ defined in \eqref{eq:mi} is positive. In addition, for any $j\in V_i$ and $0<\epsilon<m_i$,
$$
P\left(j\notin \hat{V}^{(\epsilon)}_{(i,n)}\right)\leq p_n^u := 4\exp\left\{-\frac{(m_{i}-\epsilon)^2 n^{2\xi}}{2}\right\}+\exp\left\{-O\left(n^{1/2+\xi}\right)\right\}.  
$$
3. In particular, if we take $ \epsilon_n = O (n^{ - \xi/2} ),$ where $ \xi $ is the parameter appearing in \eqref{eq:xi}, then
$$  \hat{V}^{(\epsilon_n)}_{(i,n)} =V_i  \mbox{ eventually almost surely.}$$
\end{theo}

\subsection{Results for partially observed interacting neighborhoods}
\label{App:partial_obs}
 The aim of this section is to show that
$P(j\in \hat{V}^{(\epsilon)}_{(i,n)})\to 0$ as $n\to\infty$ when $j\in F$ is functionally independent of $i$ and $V_i$ is not contained in $F$. However, it turns out that the techniques used in \cite{dglo} do not apply anymore when $V_i$ is not contained in $F$. 
For that reason, we need to work with a slightly modified version  of $\hat{V}^{(\epsilon)}_{(i,n)}$. Specifically, for a given $M>0$ and $0<\xi<1/2$,
we consider the set
\begin{equation}\label{eq:xibis}
\mathcal{T}^M_{(i,n)}= \Big\{w\in  \bigcup_{\ell =1}^{M}  C^{F\setminus\{i\},\ell} :N_{(i,n)}(w)\geq n^{1/2+\xi} \ \Big\}.
\end{equation} 
Then, for each $w\in\mathcal{T}^M_{(i,n)}$ and for any $ j \in F \setminus \{i\} $, we introduce the set
\begin{equation*}
\mathcal{T}^{M,w,j}_{(i,n)}= \Big\{v \in \mathcal{T}^M_{(i,n)}: |v|=|w|,  v_{ \{j \}^c }   = w_{ \{ j \}^c }  \Big\}
\end{equation*}
and define
\begin{equation}\label{eq:deltabis}
\Delta^M_{(i,n)}(j)=\max_{w\in \mathcal{T}^M_{(i,n)}}\max_{v \in \mathcal{T}^{M,w,j}_{(i,n)}}|\hat{p}_{(i,n)}(1|w)-\hat{p}_{(i,n)}(1|v )|.
\end{equation}
Finally, for any $\varepsilon >0,$ the new estimated interaction neighborhood of neuron $i \in F,$ given the sample $X_1(F),...,X_n(F),$ is defined as
\begin{equation}\label{eq:defestimbis}
\hat V_{(i,n)}^{(M,\varepsilon)} = \left\{ j \in F \setminus \{i\} : \Delta^M_{(i,n)}(j) > \varepsilon \right\}.
\end{equation}
It can be proved that the consistency of $\hat V_{(i,n)}^{(M,\varepsilon)}$ also holds under the same assumptions of Theorem \ref{prop:noprune}. However this goes beyond the scope of the article.

In what follows, we will show that, for any fixed $M>0$, $P(j\in \hat{V}^{(M,\epsilon)}_{(i,n)})\to 0$ as $n\to\infty$ when $j\in F$ is functionally independent of $i$ and $V_i$ is not contained in $F$. To that end, we shall introduce some extra notation.

For any subset $A$ of $I$ and integer  $\ell\geq 1$, we denote
$
C^{A,\ell}= \{ 0, 1\}^{ \{ - \ell , \ldots , -1\} \times A}.
$
For any subset $A$  of $I\setminus\{i\}$ and configuration $w\in  C^{A,\ell}$, we then define
\begin{equation}
\label{def:trans.prob}
p_{i|A}(1|w)=\lim_{t\to\infty }P(X_t(i)=1|X^{t-1}_{t-(\ell+1)}(i)=10^{\ell},X^{t-1}_{t-\ell}(A)=w).
\end{equation}
Under Assumption \ref{ass:4}, it can be shown  that $p_{i|A}(1|w)$ is well-defined and does not depend on the initial condition $X_{-\infty}^0=x\in\Omega^{adm}$. One can also show that for $w\in  C^{A,\ell},$
\begin{equation} 
\label{def:rate_of_convergence}
\hat{p}_{(i,n)}(1|w)\to p_{i|A}(1|w) \ \mbox{ in probability as} \ n\to\infty.
\end{equation} 
For the details see for instance \cite{Shields:96}.
From  \eqref{def:trans.prob}  we obtain the following proposition.
\begin{proposition}
\label{Prop:local_indep}
Let $\ell\geq 1$ and $A\subset I\setminus\{i\}$. Under Assumption \ref{ass:4}, if neuron $j\in A$ is functionally independent of $i$, then for any $w,v\in  C^{A\setminus\{i\},\ell}$ satisfying  $w_{\{j\}^c}=v_{\{j\}^c}$,
\begin{equation}
p_{i|A}(1|w)=p_{i|A}(1|v).
\end{equation} 
\end{proposition}

\begin{proof}
Without lost of generality, we assume that neuron $j$ is the only neuron in $A$ which is functionally independent of $i$. By definition, this means that for each $s\geq 1$, the variables $X_s(j)$ and $X_s(k),k\in (A\setminus\{j\})\cup\{i\}$, are independent.
As a consequence, if we write $B=A\setminus\{j\}$, $u=w_{\{j\}^c}$, $z=w^{-1}_{-\ell}(j)$ and $z'=v^{-1}_{-\ell}(j)$, it follows that  $u=v_{\{j\}^c}$ and
{\small 
\begin{eqnarray*}
\frac{P(X^{t}_{t-(\ell+1)}(i)=10^{\ell}1,X^{t-1}_{t-\ell}(B)=u|X^{t-1}_{t-\ell}(j)=z)}{P(X^{t-1}_{t-(\ell+1)}(i)=10^{\ell},X^{t-1}_{t-\ell}(B)=u|X^{t-1}_{t-\ell}(j)=z)}=\frac{P(X^{t}_{t-(\ell+1)}(i)=10^{\ell}1,X^{t-1}_{t-\ell}(B)=u|X^{t-1}_{t-\ell}(j)=z')}{P(X^{t-1}_{t-(\ell+1)}(i)=10^{\ell},X^{t-1}_{t-\ell}(B)=u|X^{t-1}_{t-\ell}(j)=z')}
\end{eqnarray*}
}
implying that
$$P(X_t(i)=1|X^{t-1}_{t-(\ell+1)}(i)=10^{\ell},X^{t-1}_{t-\ell}(A)=w)=P(X_t(i)=1|X^{t-1}_{t-(\ell+1)}(i)=10^{\ell},X^{t-1}_{t-\ell}(A)=v).$$
By taking the limit as $t\to\infty$ in the equality above, we conclude the proof.
\end{proof}


As a consequence of the Proposition \ref{Prop:local_indep}, we have the following.

\begin{proposition}
\label{prop:partial_obs}
Let $X_1 (F) ,\ldots, X_n(F)$ be a sample produced by a stochastic chain $(X_t)_{t\in\Z}$ compatible with \eqref{def:1} and \eqref{def:2}, starting from $X^{0}_{-\infty}=x $ for some fixed $ x\in \Omega^{adm}$. Under Assumption \ref{ass:4}, for any $i\in F$ and $j\in F\setminus\{i\}$ functionally independent of $i$, we have for any $M>0$ and $\epsilon>0$,
$$P\left(j\in \hat{V}^{(M,\epsilon)}_{(i,n)}\right)\to 0 \  \mbox{as} \ n\to\infty.$$
\end{proposition}
\begin{proof}
Proposition \ref{Prop:local_indep} and \eqref{def:rate_of_convergence} imply that for any $j\in F\setminus\{i\}$ functionally independent of $i$, $\ell\geq 1$ and 
$w,v\in  C^{F\setminus\{i\},\ell}$ such that $v_{\{j\}^c}=w_{\{j\}^c}$, it holds
\begin{equation}
\label{conver_in_prob}
\hat{p}_{(i,n)}(1|w)-\hat{p}_{(i,n)}(1|s)\to 0 \ \mbox{ in probability as} \ n\to\infty.
\end{equation}
Observe that for each $w\in \mathcal{T}^M_{(i,n)}$,  $|\mathcal{T}^{M,w,j}_{(i,n)}|\leq |\mathcal{T}^M_{(i,n)}|$. Since $|\mathcal{T}^M_{(i,n)}|$ is bounded by some constant depending only on $M$ and 
$$
P\left(j\in \hat{V}^{(M,\epsilon)}_{(i,n)}\right)=P\left(\cup_{w\in \mathcal{T}^M_{(i,n)}}\cup_{v \in \mathcal{T}^{M,w,j}_{(i,n)}}\{|\hat{p}_{(i,n)}(1|w)-\hat{p}_{(i,n)}(1|v )|>\epsilon\}\right),
$$
the result follows from \eqref{conver_in_prob}.
\end{proof}

\subsection{Pruning procedure}\label{sec:pruningproof}

We will show here that the pruning procedure described in section \ref{subsec:pruning} conserves the consistence property of the estimation.

We denote by  $\hat{V}^N_{(i,n)}$ the estimated neighborhood obtained after $N$ pruning steps.

We have the following result.

\begin{proposition}\label{prop:prune}

Grant Assumption \ref{ass:4}. Let $X_1 (F) ,\ldots, X_n(F)$ be a sample produced by a stochastic chain $(X_t)_{t\in\Z}$ compatible with \eqref{def:1} and \eqref{def:2}, starting from $X^{0}_{-\infty}=x $ for some fixed $ x\in \Omega^{adm}$.
Suppose additionally that $ V_i $ is finite and $ V_i \subset F .$
\\
1. {\bf (Overestimation).} \label{thm:2Iprune} For any $j\notin V_i$, we have that for any $\epsilon>0,$
$$
P\Big(j\in \hat{V}^N_{(i,n)}\Big)\leq  \left( 1- ( 1- p_n^u)^N \right)  + p_n^o,
$$
where we recall that $p_n^u$ and $p_n^o$ are defined in Theorem \ref{prop:noprune}.

2. {\bf (Underestimation).} \label{thm:2IIprune} 
Let $m_i=\inf_{u\in K_{i}}\left\{\varphi'(u)\right\}\inf_{j\in V_i}\left\{|W_{j\to i}|g_j(1)\right\}$, where $K_{i}$ is defined in \eqref{def:Kil}. Then $m_i>0$ and for any $j\in V_i$ and $0<\epsilon<m_i$,
$$
P\left(j\notin \hat{V}^N_{(i,n)}\right)\leq \left( 1- ( 1- p_n^u)^N \right)  + p_n^u.
$$
\end{proposition}

\begin{proof}
We start by proving item 1.
We denote by $F^N$ the set of neurons considered at the N-th pruning step and by $G^N:=F \setminus F^N$ the set of neurons pruned at the N-th pruning step.

We put $A^N:= \left\{ \forall k \in G^N, k \notin V_i \right\}$ for the event where all neurons pruned at the N-th pruning step are not in $V_i.$ 

For $N=1,$ we have $P\left( A^1 \right) \geq (1-p_n^u)$  by Theorem \ref{prop:noprune}, since this event corresponds exactly to the fact that the neuron pruned was not a false negative.

For $N=2,$ since $A^2 \subset A^1,$ we have
$$
P \left( A^2 \right)=P \left( A^2 / A^1 \right)P \left( A^1 \right) \geq (1-p_n^u ) P \left( A^2 / A^1 \right).
$$

When $A^1$ is realised, the realisation of $A^2$ corresponds to the fact that the second neuron pruned was not a false positive and therefore, applying Theorem \ref{prop:noprune}, we have $P \left( A^2 / A^1 \right) \geq (1-p_n^u).$ This gives us $P \left( A^2 \right) \geq (1-p_n^u)^2,$ and by recurrence, we obtain $P(A^N) \geq (1-p_n^u)^N.$ 

Now we write

$$
P \left( j\in \hat{V}^N_{(i,n)} \right)=P \left(  \left\{ j\in \hat{V}^N_{(i,n)} \right\} \cap \overline{A^N} \right) + P \left(  \left\{ j\in \hat{V}^N_{(i,n)} \right\} \cap A^N \right),
$$

where $\overline{A^N}:=\Omega \setminus A^N$ is the complementary event of $A^N.$

We have 
$$
P \left( j\in \hat{V}^N_{(i,n)} \right) \leq  P\left( \overline{A^N} \right)  + P \left( \left\{ j\in \hat{V}^N_{(i,n)} \right\} \Big/ A^N \right).
$$

Since $A^N = \left\{ V_i \subset F^N \right\},$  we can apply Theorem\ref{prop:noprune} to obtain $P \left( \left\{ j\in \hat{V}^N_{(i,n)} \right\} \Big/ A^N \right) \leq p_n^o.$

This gives us 

$$
P \left( j\in \hat{V}^N_{(i,n)} \right) \leq  \left( 1- ( 1- p_n^u)^N \right)  + p_n^o.
$$

The proof of item 2 follows the same steps.
\end{proof}

\end{appendix}
\bibliographystyle{humannat}
\bibliography{Bibli}{}

\begin{figure}[!p]
\includegraphics[scale=0.42]{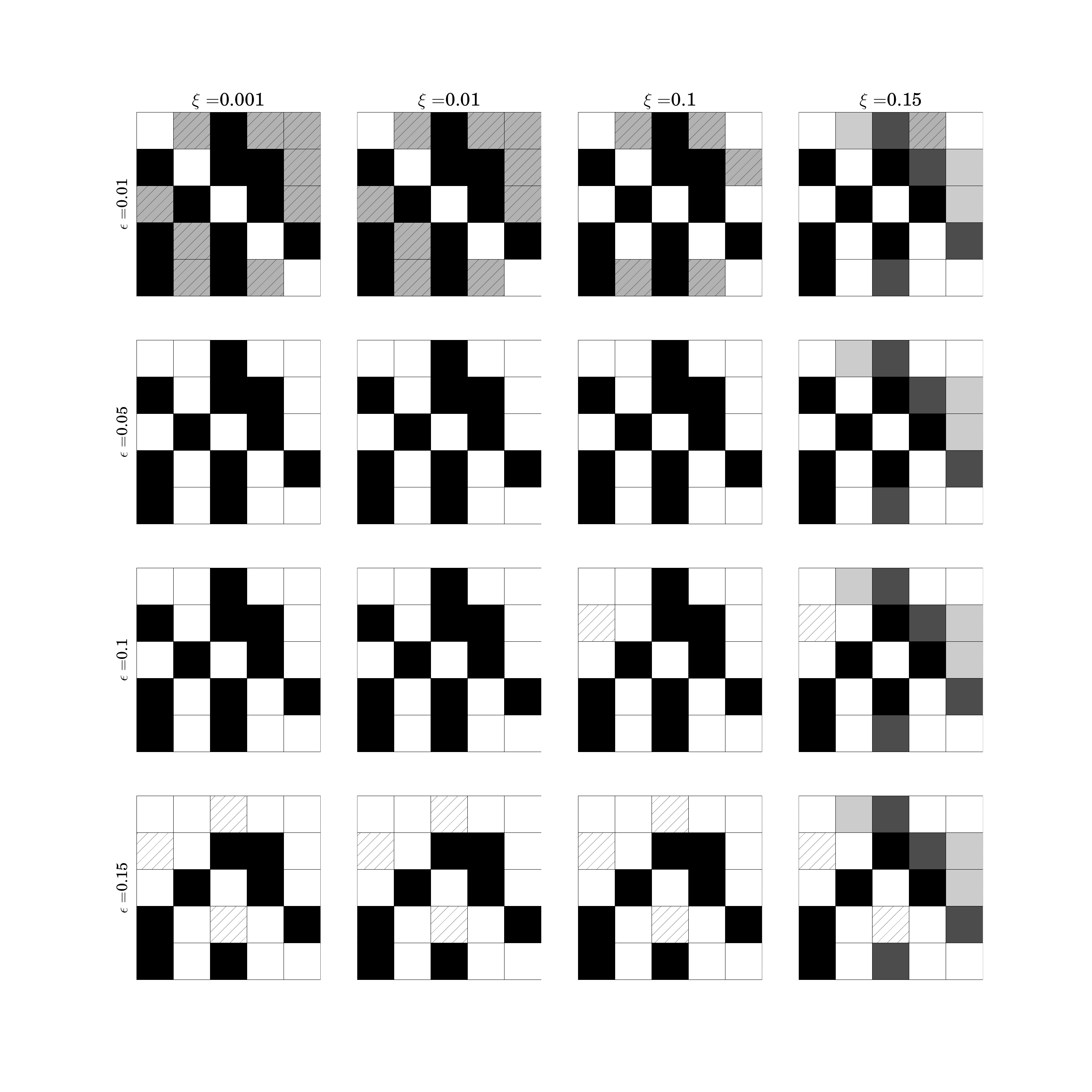}
\caption{Influence graph estimations for a simulated dataset of a GL network of 5 neurons for various values of parameters $(\epsilon,\xi)$, estimated for the same dataset with $n=10^6$. Each element in the panel is a color coded representation of connections, where rows correspond to presynaptic neurons and columns to postsynaptic neurons. Colors indicate the comparison of the estimated connectivity graph to the true connectivity matrix used in simulation to generate this dataset. Black cells correspond to a true connection correctly identified by the estimator. White ones indicate there is no connection, correctly identified as absent by the estimator. Hatched white and grey cells correspond to false negative and false positive, respectively. Grey cells  correspond to inconclusives, where there is not enough repetitions of patterns to produce an estimation. Grey cells can be of two types: dark and light grey, in order to differentiate respectively  the cases where there is or not a true connection.}
		
 \label{fig:simxieps}
\end{figure}

\begin{figure}[!p]
\centering
\makebox[\textwidth][c]{ \includegraphics[width=\linewidth]{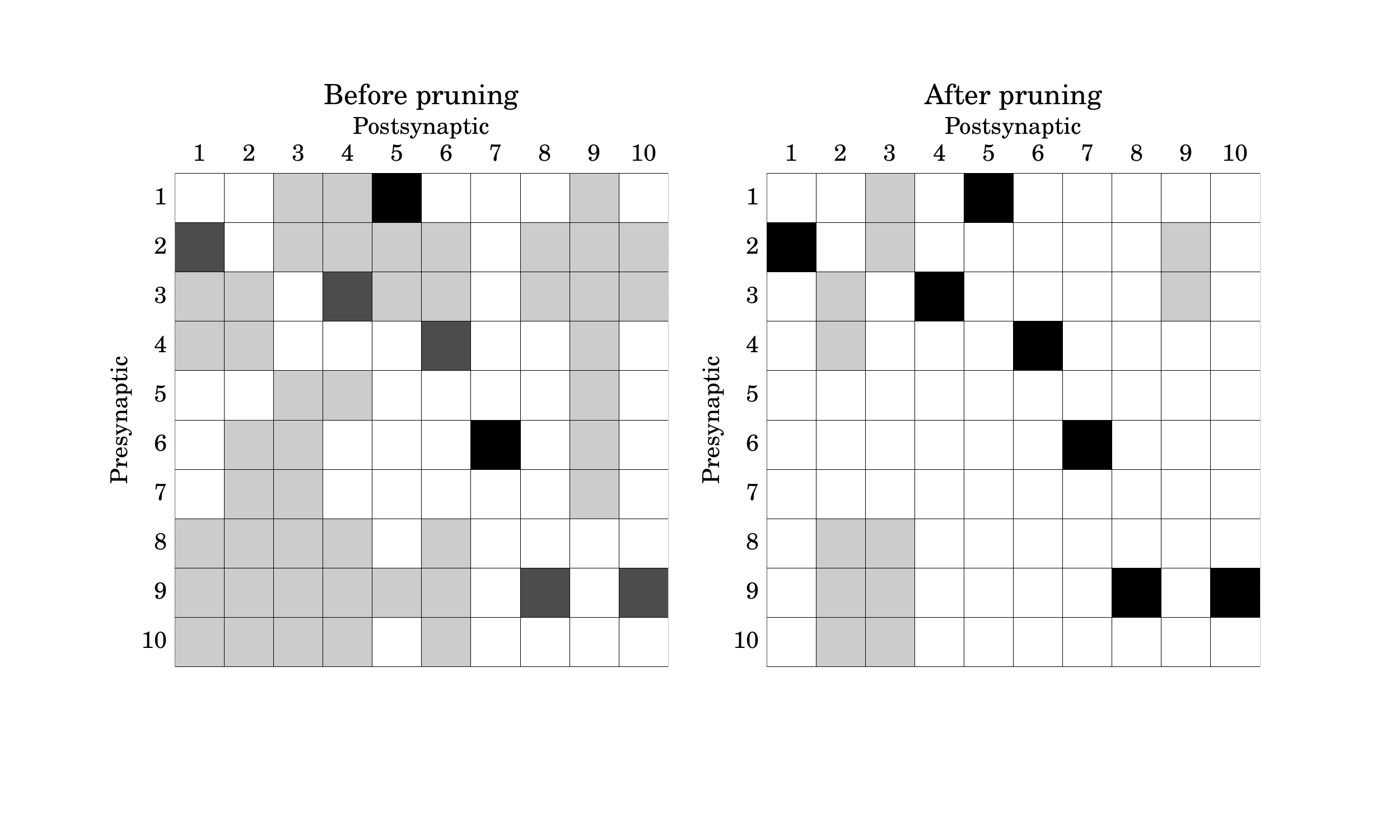}}
\caption{ Color code: Black: existent connection correctly identified by estimator. White connection: non-existent connection correctly identified as such by estimator. Dark and light grey: inconclusives corresponding respectively to existent or non-existent connections. Original estimator produces too many inconclusives {\bf (A)} for simulated dataset produced by a network of 10 neurons with $n=\num{2e5}$.  After several prunings, we obtain a closer graph estimation. The final estimation correctly identifies all existing connections for this network, but a few inconclusives remain where there are no connections.}
\label{fig:pruning}
\end{figure}

\begin{figure}[!ht]
  \includegraphics[width=1.0\linewidth]{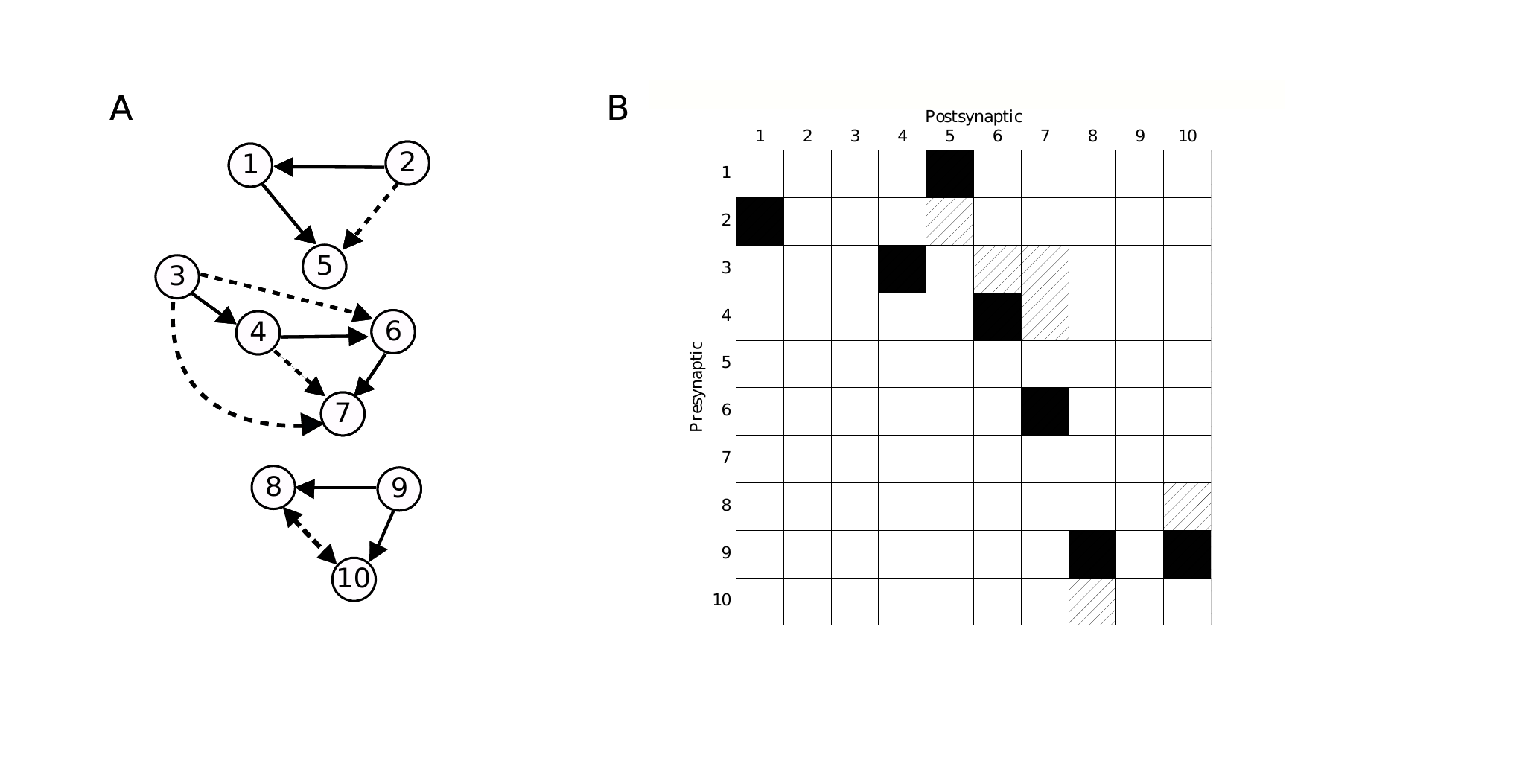}
\caption{A- Network scheme with $10$ neurons. Arrows with black lines represent direct connections. Arrows with dashed black lines represent first and second order projections. B- Complete graph recovered by procedure that identifies false positives due to projections. Black: connection accepted as true. White: connection concluded as false. Hatched: identified projections. The procedure correctly identifies true connections and projections for this network from a small sample size of $n=\num{2e5}$.}

\label{ProjMap}
\end{figure}

\begin{figure}[!p]
\centering
\makebox[\textwidth][c]{ \includegraphics[width=\linewidth]{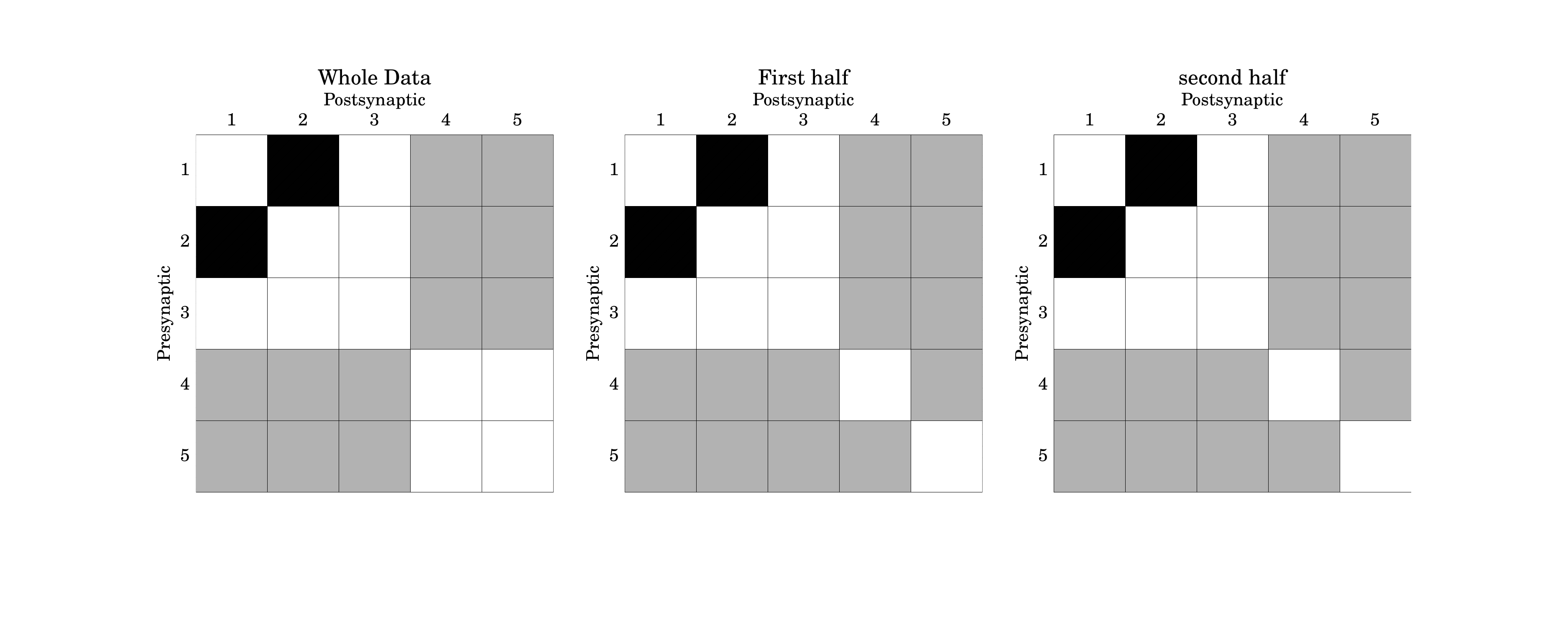}}
\caption{Graph estimation for the whole dataset (left), for the first half (middle) and the second half of the dataset (right). Black: estimated connection. White: estimator produces no connection. Grey: inconclusives}
\label{fig:truedata}
\end{figure}

\end{document}